\documentclass[a4paper, 10pt,twocolumn]{article}

\newcommand{\qed}{\rule{1.5mm}{2mm}\vspace{0.1in}}
\newenvironment{proof}{\par\noindent{\bf Proof:}}{\qed}

\usepackage{amssymb}
\usepackage{graphicx}
\usepackage{epsfig}
\newcommand{\remove}[1]{}

\newcommand{\myprob}{\Pr}

\newcommand{\ignore}[1]{}%

\newtheorem{theorem}{Theorem}[section]
\newtheorem{lemma}[theorem]{Lemma}
\newtheorem{corollary}[theorem]{Corollary}
\newtheorem{proposition}[theorem]{Proposition}

\newtheorem{claim}[theorem]{Claim}
\newtheorem{definition}[theorem]{Definition}

\begin{document}

\title{Non-Price Equilibria in Markets of Discrete Goods\\
\textmd{(working paper)}
}

\author{
Avinatan Hassidim\thanks{Google, Tel Aviv.}\and Haim
Kaplan\thanks{Google, Tel Aviv, and the School of Computer science, Tel Aviv University.
 This
research was supported in part
 by a grant from the Israel
Science Foundation (ISF), by a grant from United States-Israel Binational
Science Foundation (BSF). Email:haimk@cs.tau.ac.il.} \and Yishay Mansour\thanks{School of Computer
science, Tel Aviv University. Email:mansour@cs.tau.ac.il. This
research was supported in part
 by a grant from the the
Science Foundation (ISF), by a grant from United States-Israel Binational
Science Foundation (BSF), by a grant from the Israeli Ministry
of Science (MoS), and by the Google Inter-university center
for Electronic Markets and Auctions.}\and Noam Nisan\thanks{School of Computer Science
and Engineering, Hebrew University of Jerusalem. Supported by a
grant from the Israeli Science Foundation (ISF), and by the Google Inter-university center
for Electronic Markets and Auctions.}
%
}

\maketitle

\thispagestyle{empty}
\setcounter{page}{1}
\begin{abstract}
We study markets of indivisible items in which price-based
(Walrasian) equilibria often do not exist due to the discrete
non-convex setting.  Instead we consider Nash equilibria of the
market viewed as a game, where players bid for items, and where the
highest bidder on an item wins it and pays his bid. We first observe
that pure Nash-equilibria of this game excatly correspond to
price-based equilibiria (and thus need not exist), but that
mixed-Nash equilibria always do exist, and we analyze their
structure in several simple cases where no price-based equilibrium
exists. We also undertake an analysis of the welfare properties of
these equilibria showing that while pure equilibria are always
perfectly efficient (``first welfare theorem''), mixed equilibria
need not be, and we provide upper and lower bounds on their amount
of inefficiency.
\end{abstract}

\section{Introduction}

\subsection{Motivation}
The basic question that Economics deals with is how to ``best''
allocate scarce resources.  The basic answer is that trade can
improve everyone's welfare, and will lead to a market equilibrium: a
vector of resource prices that ``clear the market'' and lead to an
efficient allocation.  Indeed, Arrow and Debreu \cite{ArrowD54} and
much further work shows that such market equilibria exist in general
settings.

Or do they...?  An underlying assumption for the existence of
price-equilibria is always some notion of ``convexity''.  While some
may feel comfortable with the {\em restriction} to ``convex
economies'', markets of discrete items -- arguably the main object
of study in computerized markets and auctions -- are only rarely
``convex'' and indeed in most cases do {\em not} have any
price-based equilibria.  What can we predict to happen in such
markets?  Will these outcomes be efficient in any sense?  In this
paper we approach this questions by viewing the market as a game,
and studying its Nash-equilibria.

\subsection{Our Model}
To focus on the basic issue of lack of price-based equilibria, our
model does not address informational issues, assumes a single
seller, and does not assume any budget constraints.

Our seller is selling $m$ heterogeneous indivisible items to $n$
buyers who are competing for them.  Each buyer $i$ has a valuation
function $v_i$ specifying his value for each subset of the items.
I.e. for a subset $S$ of the items $v_i(S)$ specifies the value for
that buyer if he gets exactly this subset of the items, expressed in
some currency unit (i.e., the buyers are quasi-linear).  We will
assume free disposal, i.e., that the $v_i$'s are monotonically
non-decreasing, but nothing beyond that.

The usual notion of price-based equilibrium in this model is called
a Walrasian equilibrium: a set of item prices $p_1 \ldots p_m$ and a
partition $S_1\ldots S_n$ of the $m$ items among the $n$ buyers such
that each buyer gets his ``demand'' under these prices, i.e., $S_i
\in argmax_S (v_i(S)-\sum_{j \in S}p_j)$.  When such equilibria
exist they maximize social welfare, $\sum_i v_i(S_i)$, but
unfortunately it is known that they only rarely exist -- exactly
when the associated integer program has no integrality gap (see
\cite{BlunrosenNisan07} for a survey).

We will consider this market situation as a {\em game} where each
player\footnote{We use interchangeably the terms: player, bidder and
buyer, and all three have the same meaning.} $i$ announces $m$
offers $b_{i1}, \ldots b_{im}$, with the interpretation that
$b_{ij}$ is player $i$'s bid of item $j$. After the offers are made,
$m$ independent first price auctions are being made. That is the
utility of each bidder $i$ is given by $u_i(b) = v_i(S_i) - \sum_{j
\in S_i} b_{ij}$ where $S_1 \ldots S_n$ are a partition of the $m$
items with the property that each item went to a highest bidder on
it.  Some care is needed in the case of ties -- two (or more)
bidders $i \ne i'$ that place the highest bid $b_{ij}=b_{i'j}$ for
some item $j$. In this case a {\em tie breaking rule} is needed to
complete the specification of the allocation and thus of the game.
Importantly, we view this as a {\em game with complete information},
so each player knows the (combinatorial) valuation function of each
other player.

\subsection{Pure Nash Equilibrium}
Our first observation is that the pure equilibria of this game capture
exactly the Walrasian equilibria of the market.  This justifies our
point of view that when we later allow mixed-Nash equilibria as
well, we are in fact strictly generalizing the notion of
price-equilibria.

\vspace{0.1in} \noindent {\bf Theorem:} Fix a profile of valuations.
Walrasian equilibria of the associated market are in 1-1
correspondence with pure Nash equilibria of the associated game.
This holds in the exact sense for {\em some} tie-breaking rule, and
holds in the sense of limits of $\epsilon$-Nash equilibria for {\em
all} ties-breaking rules. \vspace{0.1in}

A profile of strategies (bids) in the game is called a ``limit of
$\epsilon$-Nash equilibria'' if for every $\epsilon>0$ there exists
a sequence of $\epsilon$-Nash equilibria that approach it.

Let us demonstrate this theorem with a trivial example: a single
item on sale and two bidders who have values of 1 and 2 respectively
for it.  A Walrasian equilibrium can fix the item's price $p$
anywhere between 1 and 2, at which point only the second bidder
desires it and the market clears. In the associated game (with any
tie breaking rule), a bid $p$ for the first player and bid
$p+\epsilon$ for the second player will be an
$\epsilon$-Nash-equilibrium.  In the special case that the tie
breaking rule gives priority to the second bidder, an exact
pure-Nash equilibrium will have both bidders bidding $p$ on the
item.

This theorem is somewhat counter intuitive as strategic (non-price-taking)
buyers in markets may improve their utility by strategically
``reducing demand''.  Yet, in our setting strategic buyers still
reach the basic non-strategic price-equilibrium.

As an immediate corollary of the fact that a Walrasian equilibrium
optimizes social welfare (``The first welfare theorem''), we get the
same optimality in our game setting:

\vspace{0.1in} \noindent {\bf Corollary -- A ``First Welfare
Theorem''} For every profile of valuations and every tie-breaking
rule, every pure Nash equilibrium of the game (including a limit of
$\epsilon$-equilibria) optimizes social welfare.  In other words,
the price of anarchy of pure Nash equilibria is trivial.
\vspace{0.1in}

\subsection{Mixed Nash Equilibria}
As mentioned above, since Walrasian equilibria only rarely exists,
so do only rarely exist pure Nash equilibria in our games.  So it is
quite natural to consider also the standard generalization,
Mixed-Nash equilibria of our market games.  The issue of existence
of such mixed Nash equilibria is not trivial in our setting as
buyers have a continuum of strategies and discontinuous utilities so
Nash's theorem does not apply.  Nevertheless, there has been a
significant amount of economic work on these types of settings and a
theorem of Simon and Zame \cite{SimonZame90} provides at least a partial
general positive answer:

\vspace{0.1in} \noindent {\bf Corollary} (to a theorem of \cite{SimonZame90}):
For every profile of valuations, there exists some (mixed)
tie-breaking rule such that the game has a mixed-Nash equilibrium.
\vspace{0.1in}

It seems that, like in the case of pure equilibria, an
$\epsilon$-Nash equilibrium should exist for all tie breaking rules,
but we have not been able to  establish this.

Once existence is established, we turn our attention towards
analyzing what these mixed equilibria look like.  We start with the
two basic examples that are well known not to have a price
equilibria:

\vspace{0.1in} \noindent {\bf Example -- Complements and Substitutes
Bidders:} In this example there are two items and two bidders.  The
first bidder (``OR bidder'') views the two items as perfect
substitutes and has value of $v_{or}$ for either one of them (but is
not interested in getting both).  The second bidder (``AND bidder'')
views them as complements and values the bundle of both of them at
$v_{and}$ (but is not interested in either of them separately).  It
is not difficult to see that when $v_{and} < 2v_{or}$ no pure
equilibrium exists, however we find specific distributions $F$ and
$G$ for the bids of the players that are in mixed-Nash equilibrium.

\vspace{0.1in} \noindent {\bf Example -- Triangle:} In this example
there are three items and three players.  Each of the players is
interested in a specific pair of items, and has value 1 for that
pair, and 0 for any single item, or any other pair.  A pure Nash
equilibrium does not exist, but we show that the following is a
mixed-Nash equilibrium: each player picks a bid $x$ uniformly at
random in the range $[0,1/2]$ and bids this number on each of the
items. Interestingly the expected utility of each player is zero.
We generalize the analysis to the case of single minded players,
each desiring a set of size $k$, each item is desired by $d$
players, and no two players' sets intersect in at most a single
item.
\vspace{0.1in}

We generalize our analysis to more general examples of these veins.
In particular, these provide examples where the mixed-Nash
equilibrium is not optimal in terms of maximizing social welfare and
in fact is far from being so.

\vspace{0.1in} \noindent {\bf Corollary -- A ``First Non-Welfare
Theorem'':} There are profiles of valuations where a mixed-Nash
equilibrium does not maximize social welfare.  There are examples
where pure equilibria (that maximize social welfare) exist and yet a
mixed Nash equilibrium achieves only $O(1\sqrt{m})$ fraction of
social welfare (i.e., ``price of anarchy'' is $\Omega(\sqrt{m})$).
There exist examples where all mixed-Nash equilibria achieve at most
$O(\sqrt{(\log m)/m})$ fraction of social welfare (i.e., ``price of
stability'' is $\Omega(\sqrt{m/(\log m)})$).

At this point it is quite natural to ask how much efficiency can be
lost, in general, as well for interesting subclasses of valuations,
which we answer as follows.

\vspace{0.1in} \noindent {\bf Theorem -- An ``Approximate First
Welfare Theorem'':} For every profile of valuations, every
tie-breaking rule, and every mixed-Nash equilibrium of the game we
have that the expected social welfare obtained at the equilibrium is
at least $1/\alpha$ (the ``price of anarchy'') times the optimal
social welfare, where
\begin{enumerate}
\item $\alpha \le 2\beta$ if all valuations $\beta$-fractionally subadditive.
(The case $\beta=1$ correponds to fractionally subadditive valuations, also known
as XOS valuations. They include the set of sub-modular valuations.)
\item $\alpha = O(\log m)$ if all valuations are sub-additive.
\item $\alpha = O(m)$, in general.
\end{enumerate}
These bounds apply also to correlated-Nash equilibria and even to
coarse-correlated equilibria. \vspace{0.1in}

A related PoA result is that of \cite{BhawalkarR11} which derive PoA
for $\beta$-fractionally sub-additive bidders in a second price
simultaneous auction under the assumption of ``conservative
bidding.'' In this work we use the first price (rather than the
second price) and do not make any assumption regarding the bidding.

 \ignore{ The first
two parts of this theorem are somewhat similar to the analysis of
\cite{BhawalkarR11} who prove similar bounds when the payments are
according to the second-price (rather than the bids themselves), and
under the additional assumption of ``conservative bidding'' -- an
assumption which we do not require.}

Finally  we extend these results also
 to a Bayesian setting where players have only partial information
on the valuations of the other players. We show that for any prior
distribution on the valuations and in every Bayesian Nash
equilibrium, where each player bids only based on his own valuation
(and the knowledge of the prior), the average social welfare is
lower by at most $\alpha=O(mn)$ than the optimal social welfare
achieved with full shared knowledge and cooperation of the players.
For a prior which is a product distribution over valuations which
are $\beta$-fractionally sub-additive we show that $\alpha =
4\beta$, which implies a bound of $4$ for sub-modular valuations and
a bound of $O(\log m)$for sub-additive valuations. Our proof
methodology for this setting is similar to that of
\cite{BhawalkarR11}.

\subsection{Open Problems and Future Work}
We consider our work as a first step in the systematic study of notions
of equilibrium in markets where price equilibria do not exist.  Our
own work focused on the mixed-Nash equilibrium,  its existence and
form, and its welfare properties.  It is certainly natural to
consider other properties of such equilibria such as their revenue
or invariants over the set of equilibria.  One may also naturally
study other notions of equilibrium such as those corresponding
outcomes of natural dynamics (e.g. coarse correlated equilibria
which are the outcome of regret minimization dynamics).  It is also
natural to consider richer models of markets (e.g. two-sided ones,
non-quasi-linear ones, or ones with partial information).

Even within the modest scope of this paper, there are several
remaining open questions: the existence of mixed-Nash equilibrium
under any tie-breaking rule; the characterization of all equilibria
for the simple games we studied; and closing the various gaps in our
price of anarchy and price of stability results.

\section{Model}

We have a set $M$ of $m$ heterogeneous indivisible items for sale to
a set $N$ of $n$ bidders.  Each bidder $i$ has a valuation function
$v_i$ where for a set of items $S \subseteq M$, $v_i(S)$ is his
value for receiving the set $S$ of items.  We will not make any
assumptions on the $v_i$'s except that they are monotone non
decreasing (free disposal) and that $v_i(\emptyset)=0$. We assume
that the utility of the bidders is quasi-linear, namely, if bidder
$i$ gets subset $S_i$ and pays $p_i$ then
$u_i(S_i,p_i)=v_i(S_i)-p_i$.

We will consider this market situation as a game where the items are
sold in simultaneous first price auctions. Each bidder $i\in N$
places a bid $b_{ij}$ on each each item $j\in M$, and the highest
bidder on each item gets the item and pays his bid on the item.  We
view this as a {\em game with complete information}.  The utility of
each bidder $i$ is given by $u_i(b) = v_i(S_i) - \sum_{j \in S_i}
b_{ij}$ where $S_1 ... S_n$ are a partition of the $m$ items with
the property that each item went to the bidder that gave the highest
bid for it.

Some care is required in cases of ties, i.e., if for some bidders $i
\ne i'$ and  an item $j\in M$ we have that $b_{ij}=b_{i'j}$ are both
highest bids for item $j$.  In these cases the previous definition
does not completely specify the allocation, and to complete the
definition of the game we must specify a tie breaking rule that
chooses among the valid allocations. (I.e. specifies the allocation
$S_1, \ldots , S_n$ as a function of the bids.) In general we allow
any tie breaking rule, a rule that may depend arbitrarily on all the
bids.  Even more, we allow randomized (mixed) tie breaking rules in
which some distribution over deterministic tie breaking rules is
chosen. We will call any game of this family (i.e.,with any tie
breaking rule) a ``first price simultaneous auction game'' (for a
given profile of valuations).

\section{Pure Nash Equilibrium}
\ignore{ We are now dealing with $m$ heterogeneous indivisible items
for sale among $n$ bidders.  Each bidder $i$ has a valuation
function $v_i$ where for a set of items $S \subseteq \{1...m\}$,
$v_i(S)$ is his value for receiving the set $S$ of items.  We will
not make any assumptions on the $v_i$'s except that they are
monotone non decreasing and that $v_i(\emptyset)=0$.}

The usual analysis of this scenario considers a market situation and
a price-based equilibrium:
\begin{definition}
A partition of the items $S_1 ... S_n$ and a non-negative vector of
prices $p_1 ... p_m$ are called a Walrasian equilibrium if for every
$i$ we have that $S_i \in argmax_S (v_i(S)-\sum_{j \in S} p_i)$.
\end{definition}
\ignore{ {\bf Comment:} In our definition we require that the
$S_i$'s are a partition; an equivalent definition would also allow
un-allocated items as long as they are priced at zero. {\tt [[YM:
with free disposal this is redundant]]} }

We consider bidders participating in a  simultaneous first price
auction game, with some tie breaking rule.

Our first observation is that pure equilibria of a first price
simultaneous auction game correspond to Walrasian equilibiria of the
market.  In particular the price of anarchy of pure equilibria is 1.
\begin{proposition} \label{proposition:pure}
\begin{enumerate}
\item
A profile of valuation functions $v_1 ... v_n$ admits a Walrasian
equilibrium with given prices and allocation if and only if the first price
simultaneous auction game
for these valuations  has a pure Nash equilibrium for some tie
breaking rule with these winning prices and allocation.
\item
Every pure Nash equilibrium of a first price simultaneous auction
game achieves optimal social welfare.
\end{enumerate}
\end{proposition}
\begin{proof}
Let $S_1, \ldots , S_n$ and $p_1, \ldots , p_m$ be a Walrasian
equilibrium. Consider the bids where $b_{ij}=p_j$ for all $j$ and
let the game break ties according to $S_1...S_n$.  Why are these
bids a pure equilibrium of this game?  Since we are in a Walrasian
equilibrium, each player gets a best set for him under the prices
$p_j$.  In the game, given the bids of the other players, he can
never win any item for strictly less than $p_j$, whatever his bid,
and he does wins the items in $S_i$ for price $p_j$ exactly, so his
current bid is a best response to the others\footnote{The reader may
dislike the fact that the bids of loosing players seem artificially
high and indeed may be in weakly dominated strategies. This however
is unavoidable since, as we will see in the next section,
counter-intuitively sometimes there are no pure equilibria in
un-dominated strategies.  What can be said is that minimal Walrasian
equilibria correspond to pure equilibria of the game with strategies
that are limits of un-dominated strategies.}.

Now fix a pure Nash equilibrium of the game with a given tie
breaking rule.  Let $S_1, \ldots , S_n$ the allocation specified by
the tie breaking rule, and let $p_j = \max_i b_{ij}$ for all $j$. We
claim that this is a Walrasian equilibrium.  Suppose by way of
contradiction that some player $i$ strictly prefers another bundle
$T$ under these prices.  This contradicts the original bid of $i$
was a best reply since the deviation bidding $b_{ij}=p_j+\epsilon$
for $j \in T$ and $b_{ij}=0$ for $j \not\in T$ would give player $i$
the utility from $T$ (minus some $\epsilon$'s) which would be more
than he currently gets from $S_i$ -- a contradiction.

The allocation obtained by the game, is itself the allocation in a
Walrasian equilibrium, and thus by the First Welfare Theorem is a
social-welfare maximizing allocation.
\end{proof}

Two short-comings of this proposition are obvious: first is the
delicate dependence on tie-breaking: we get a Nash equilibrium only
for some, carefully chosen, tie breaking rule. In the next section
we will show that this is un-avoidable using the usual definitions,
but that it is not a ``real'' problem: specifically we show that for
any tie-breaking rule we get arbitrarily close to an equilibrium.

The second short-coming is more serious: it is well known that
Walrasian equilibria exist only for restricted classes of valuation
profiles\footnote{When all valuations are ``substitutes''.}.  In the
general case, there is no pure equilibrium and thus the result on
the price of anarchy is void. In particular, the result does {\em
not} extend to mixed Nash equilibria and in fact it is not even
clear whether such mixed equilibria exist at all since Nash's
theorem does not apply due to the non-compactness of the space of
mixed strategies.  This will be the subject of the the following
sections.

\subsection{Tie Breaking and Limits of $\epsilon$-Equilibria}

This subsection shows that the quantification to some tie-breaking rule
in the previous theorem is unavoidable. Nevertheless we argue that
it is  really just a technical issue since
 we can show that for every tie breaking rule there is
a limit of $\epsilon$-equilibria.

\subsubsection*{A first price auction with the wrong tie breaking
rule}
Consider the full information game describing a first price auction
of a single item between Alice, who has a value of 1 for the item,
and Bob who values it at 2, where the bids, $x$ for Alice and $y$
for Bob, are allowed to be, say, in the range $[0,10]$. The full
information game specifying this auction is defined by $u_A(x,y)=0$
for $x<y$ and $u_A(x,y) = 1-x$ for $x>y$, and $u_B(x,y)=2-y$ for
$x<y$ and $u_B(x,y)=0$ for $x>y$. Now comes our main point: how
would we define what happens in case of ties?  It turns out that
formally this ``detail'' determines whether a pure Nash equilibrium
exists.

Let us first consider the case where ties are broken in favor of
Bob, i.e., $u_B(x,y)=2-y$ for $x=y$ and $u_A(x,y)=0$ for $x=y$.  In
this case one may verify that $x=1,y=1$ is a pure Nash
equilibrium\footnote{The bid $x=1$ is weakly dominated for Alice.
Surprisingly, however, there is no pure equilibrium in un-dominated
strategies: suppose that some $y$ is at equilibrium with an
un-dominated strategy $x<1$. If $y ge 1$ then reducing $y$ to $y=x$
would still make Bob win, but at a lower price.  However, if $y<1$
too, then the loser can win by bidding just above the current winner
-- contradiction.}.

Now let us look at the case that ties are broken in favor of Alice,
i.e $u_A(x,y)= 1-x$ and $u_B(x,y)=0$ for $x=y$. In this case no pure
Nash equilibrium exists: first no $x \ne y$ can be an equilibrium
since the winner can always reduce his bid by $\epsilon$ and still
win, then if $x=y>1$ then Alice would rather bid $x=0$, while if
$x=y<2$ then Bob wants to deviate to $y + \epsilon$ and to win,
contradiction.

This lack of pure Nash equilibrium doesn't seem to capture the
essence of this game, as in some informal sense, the "correct" pure
equilibrium is $(x=1, y=1+\epsilon)$ (as well as $(x=1-\epsilon,
y=1)$), with Bob winning and paying $1+\epsilon$ ($1$).  Indeed
these are $\epsilon$-equilibria of the game.  Alternatively, if we
discretize the auction in any way allowing some minimal $\epsilon$
precision then bids close to $1$ with minimal gap would be a pure
Nash equilibrium of the discrete game. We would like to formally
capture this property: that $x=1$, $y=1$ is arbitrarily close
 to an equilibrium.

\subsubsection*{Limits of $\epsilon$-Equilibria}
We will become quite abstract at this point and consider general
games with (finitely many) $n$ players whose strategy sets may be
infinite.  In order to discuss closeness we will assume that the
pure strategy set $X_i$ of each player $i$ has a metric $d_i$ on it.
In applications we simply consider the Euclidean distance.
\begin{definition} $(x_1 ... x_n)$ is called a {\em limit
(pure) equilibrium} of a game $(u_1 ... u_n)$ if it is the limit
of $\epsilon$-equilibria of the game, for every $\epsilon>0$.
\end{definition}

Thus in the example of the first price auction, $(1,1)$ is a
limit equilibrium, since for every $\epsilon>0$, $(1,1+\epsilon)$
is an $\epsilon$-equilibrium.  Note that if all the $u_i$'s are
continuous at the point $(x_1...x_n)$ then it is a limit
equilibrium only if it is actually a pure Nash equilibrium.  This,
in particular, happens everywhere if all strategy spaces are
discrete.

We are now ready to state a version of the previous proposition that
is robust to the tie breaking rule:
\begin{proposition}
\begin{enumerate}
\item
For every first price simultaneous auction game with any tie
breaking rule, a profile of valuation functions $v_1 ... v_n$ admits
a Walrasian equilibrium with given prices and allocation
if and only if the game has a limit Nash
equilibrium for these valuations with these winning prices and allocation.
\item
Every limit Nash equilibrium of a first price simultaneous
auction game achieves optimal social welfare.
\end{enumerate}
\end{proposition}
\begin{proof}
Let $S_1...S_n$ and $p_1 ... p_m$ be a Walrasian equilibrium.
Consider the bids where $b_{ij}=p_j+\epsilon$ for all $j \in S_i$
and $b_{ij}=p_j$ for all $j \not\in S_i$.  Why are these bids an
$m\epsilon$-equilibrium of this game?  Since we are in a Walrasian
equilibrium, each player gets a best set for him under the prices
$p_j$.  In the game, given the bids of the other players, he can
never win any item for strictly less than $p_j$, whatever his bid,
and player $i$ does win each item $j$ in $S_i$ for price
$p_j+\epsilon$, so his current bid is a best response to the others
up to an additive $\epsilon$ for each item he wins.

Now fix a limit Nash equilibrium $(b_{ij})$ of the game with some
tie breaking rule and let $(b'_{ij})$ be an $\epsilon$-equilibrium
of the game with $|b_{ij}-b'_{ij}| \le \epsilon$ for all $i,j$ and
with no ties; let $S_1 ... S_n$ the allocation implied; and let $p_j
= \max_i b_{ij}$ for all $j$.  We claim that this is an
$m\epsilon$-Walrasian equilibrium.  Suppose by way of contradiction
that for some player $i$ and some bundle $T \ne S_i$,
$v_i(T)-\sum_{j \in T} p_j > v_i(S_i)-\sum_{j \in S_i}p_j +
m\epsilon$.  This would contradict the original bid of $i$ being an
$\epsilon$-best reply since the deviation bidding
$b_{ij}=p_j+\epsilon$ for $j \in T$ and $b_{ij}=0$ for $j \not\in T$
would give player $i$ the utility from $T$ up to $m\epsilon$ which
would be more than he currently gets from $S_i$ -- a contradiction.

Now let $\epsilon$ approach zero and look at the sequence of price
vectors $\vec{p}$ and sequence of allocations obtained as
$(b'_{ij})$ approach $(b_{ij})$.  The sequence of price vectors
converges to a fixed price vector (since they are a continuous
function of the bids).  Since there are only a finite number of
different allocations, one of them appears infinitely often in the
sequence.  It is now easy to verify that this allocation with the
limit price vector are a Walrasian equilibrium.
\end{proof}

\section{General Existence of Mixed Nash Equilibrium}

In this section we ask whether such a game need always even have a
mixed-Nash equilibrium.  This is not a corollary of Nash's theorem
due to the continuum of strategies and discontinuity of the
utilities, and indeed even zero-sum two-player games with $[0,1]$ as
the set of pure strategies of each player may fail to have any
mixed-Nash equilibrium or even an $\epsilon$-equilibrium\footnote{A
well known example is having highest bidder win, as long as his bid
is strictly less than 1, in which case he looses (with ties being
ties).}.  There is some economic literature about the existence of
equilibiria in such games (starting e.g. with \cite{dasguptamaskin86, fundenberglevine86}), and a
theorem of Simon and Zame \cite{SimonZame90}, implies that for {\em some}
(randomized) tie breaking rule, a mixed-Nash equilibrium exists. The
main example of their (more general) theorem is the following (cf.
page 864):

Suppose we are given strategy spaces $S_i$, a dense subset $S^*$ of
$S = S_1 \times \cdots \times S_n$, and a bounded continuous
function $\varphi : S^* \rightarrow \Re^n$.  Let $C_\varphi : S
\rightarrow \Re^n$ be the correspondence whose graph is the closure
of the graph of $\varphi$, and define $Q_\varphi(s)$ to be the
convex hull of $C_\varphi(s)$ for each $s \in S$.  We call the
correspondence $Q_{\varphi}$ the convex completion of $\varphi$.
These are Simon and Zame's motivating example of ``games with an
endogenous sharing rule'', and their main theorem is that these have
a ``solution'': a pair $(q, \alpha)$, where $q$ is a ``sharing
rule'', a Borel measurable selection from the payoff correspondence
$Q$ and $\alpha = (\alpha_1 , \ldots , \alpha_n)$ is a profile of
mixed strategies with the property that each player's action is a
best response to the actions of other players, when utilities are
according to the sharing rule $q$.

Now to how this applies in our setting: $S^*$ will be the set of
bids with no ties, i.e., where for all $j$ and all $i \ne i'$ we
have that $b_{ij} \ne b_{i'j}$, which is clearly dense (since bids
with ties have measure zero).  Here $\varphi$ is simply the vector
of utilities of the players from the chosen allocation which is
fully determined and continuous in $S^*$ -- when there are no ties.
For $b \not\in S^*$, we have that $C_\varphi(\vec{b})$ is the set of
utility vectors obtained from all possible deterministic
tie-breaking rules at $\vec{b}$ (each of which may be obtained as a
limit of bids with no ties), and $Q_\varphi$ is the set of mixtures
(randomizations) over these.  The solution thus provides a
randomized tie-breaking rule $q$ and mixed strategies that are a
mixed-Nash equilibrium for the game with this tie-breaking rule.  So
we get:
\begin{corollary}
The first price simultaneous auction game for any profile of
valuations has a mixed-Nash equilibrium for some randomized
tie-breaking rule.
\end{corollary}
We suspect that the tie-breaking rule is not that significant and
that mixed $\epsilon$-Nash-equilibria (or maybe even exact
Nash-equilibria) actually exist for {\em every} tie-breaking rule,
similarly to the case of pure equilibria in this paper, or as in the
somewhat related setting of \cite{JacksonsSwinkels05} where an ``invariance'' in
the tie-breaking rule holds.

\section{Mixed-Nash Equilibria: Examples}

In this section we study some of the simplest examples
of markets in our setting that do not have a Walresian
equilibrium.

\subsection{The AND-OR Game}
\label{section-and-or}

We have two players an {\tt AND} player and {\tt OR} player. The
{\tt AND} player has a value of $1$ if he gets all the items in $M$,
and the {\tt OR} player has a value of $v$ if he gets any item in
$M$. Formally, $v_{and}(M)=1$ and for $S\neq M$ we have
$v_{and}(S)=0$, also, $v_{or}(T)=v$ for $T\neq\emptyset$ and
$v_{or}(\emptyset)=0$.

When $v\leq 1/m$ there is a Walresian equilibrium with a price
of $v$ per item. By Proposition \ref{proposition:pure} this implies a pure Nash Equilibrium
in which both
players bid $v$ on each item, and the {\tt AND} player wins all the
items. Therefore, the interesting case is when $v>1/m$.
It is easy to verify that in this case
is no Walresian equilibrium.
We start with the case that $|M|=2$ and later extend it to the case
of arbitrary size.
Here is a mixed Nash equilibrium for two items.
\begin{itemize}

\item The AND player bids $(y,y)$ where $0 \le y \le 1/2$ according to cumulative distribution
$F(y)=(v-1/2)/(v-y)$  (where $F(y)=Pr[bid \le y]$).
In particular, There is an atom at 0: $Pr[y=0]=1-1/(2v)$.

\item The OR player bids $(x,0)$ with probability 1/2 and $(0,x)$ with probability 1/2,
where $0 \le x \le 1/2$ is distributed according to cumulative
distribution  $G(x)=x/(1-x)$.

\end{itemize}

Note that since the OR player does not have any mass points in his
distribution, the equilibrium would apply to any tie breaking rule.

We start by defining a {\em restricted AND-OR} game, where the AND
player must bid the same value on both items, and show that the
above strategies are a mixed Nash equilibrium for it.



\begin{claim}
Having the {\tt AND} player bid using $F$ and the {\tt OR} player
bid using $G$ is a mixed Nash equilibrium of the restricted AND-OR
game for two items.
\end{claim}

\begin{proof}
 Let us compute the expected utility of the AND player from some pure bid $(y,y)$.  The AND player wins one item for sure, and wins the second
 item too if $y>x$, i.e., with probability $G(y)$.  If he wins a single item he pays $y$, and he wins both items he pays $2y$.  His expected utility
 is thus $G(y)(1-y)-y = 0$ for any $0 \le y \le 1/2$ (and is certainly negative for $y>1/2$).  Thus any $0 \le y \le 1/2$ is a
 best-response to the OR player.

Let us  compute the expected utility of the OR player from the pure
bid $(0,x)$ (or equivalently $(x,0)$).  The OR player wins an item
if $x>y$, i.e., with probability $F(x)$, in which case he pays $x$,
for a total utility of $(v-x) \cdot F(x) = v-1/2$, for every $0 \le
x \le 1/2$ (and $x>1/2$ certainly gives less utility). Thus any $0
\le x\le 1/2$ is a  best-response to the AND player.
\end{proof}

Next we generalize the proof to the unrestricted setting.

\begin{theorem}
Having the {\tt AND} player bid using $F$ and the {\tt OR} player
bid using $G$ is a mixed Nash equilibrium of the AND-OR game for two
items.
\end{theorem}

\begin{proof}
We first show that if the AND player plays the mixed strategy $F$
then $G$ is a best response for the OR player.
This holds since when the AND player is playing $F$, then all its
bids are of the form $(y,y)$ for some $y \in[0,,1/2]$. Any bid
$(x_1,x_2)$ of the OR player, with $x_1 \le x_2$, is dominated by
$(0, x_2)$, since the AND player is  restricted to bidding $(y,y)$.
Therefore, $G$ is a best response for the OR player.

We now need to show that if the OR player plays the mixed strategy
$G$ then $F$ is a best response for the AND player.


Let $Q(x,y)$ be the cumulative probability of the OR player, i.e.,
 \[
 Q(x,y)= \Pr[bid_1 < x , bid_2<y]
 =\frac{x}{2(1-x)}+\frac{y}{2(1-y)}.
 \]
for $x,y\in [0,\frac{1}{2}]$. The AND utility function, given its
distribution $P$, is:
\[
U_{AND} = E_{(x,y)\sim P}[u_{and}(x,y)],\] where \[u_{and}(x,y)=
1\cdot Q(x,y) - \left(x Q(x,1)+ y Q(1,y)\right)
\]

We show that for any $x,y \in [0,\frac{1}{2}]$ we have
$u_{and}(x,y)=0$. This follows since,
%
\begin{eqnarray*}
u_{and}(x,y) &= &1\cdot Q(x,y) - \left(x Q(x,1)+y Q(1,y)\right)\\
&= &\left(\frac{x}{2(1-x)}+\frac{y}{2(1-y)} \right)-
x\left(\frac{x}{2(1-x)}+\frac{1}{2}\right)\\
& & - y \left(\frac{1}{2}+ \frac{y}{2(1-y)}\right)\\
&=& (1-x)\frac{x}{2(1-x)}+(1-y)\frac{y}{2(1-y)} - \frac{x}{2}-\frac{y}{2}\\
& = & 0,
\end{eqnarray*}
which completes the proof.
\end{proof}

We now extend the result to the AND-OR game with $m$ items.
%
The AND player selects $y$ using the cumulative probability
distribution $F(y)=\frac{v-\frac{1}{m}}{v-y}$ for $y\in [0,1/m]$, and
 bids $y$ on all the items.
The OR player selects $x$ using the cumulative probability
distribution  $G(x)=\frac{(m-1)x}{(1-x)}$, where $x\in[0,1/m]$, and an
$i$ uniformly from $M$, and bids $x$ on item $i$ and zero on all the
other items.

\begin{theorem}
Having the AND player bid using $F$ and the OR player with $G$ is a
mixed Nash equilibrium.
\end{theorem}
\begin{proof}
Let $Q(x)$, for $x\in [0,1/m]^m$ be the  cumulative probability
distribution of the bids of
the OR player. Given that the OR player bids using $G$ it follows that
\[
 Q(x)= \Pr[\forall i \;bid_i < x_i ]
 =\sum_{i=1}^m \frac{x_i}{1-x_i}\left(\frac{m-1}{m}\right)
 \]
for $x\in [0,\frac{1}{m}]^m$.
Let $P$ denote the
cumulative probability
distribution of the bids of
the AND player. Then the utility of the AND player
is:
\[
U_{AND} = E_{x\sim P}[u_{and}(x)],\] where \[u_{and}(x)= 1\cdot Q(x)
- \sum_{i=1}^m x_i Q(x_i,(1/m)_{-i}) \ .
\]
We show that for any $x \in [0,\frac{1}{m}]^m$ we have
$u_{and}(x)=0$.
\begin{eqnarray*}
u_{and}(x) &= &1\cdot Q(x) - \left(\sum_{i=1}^k x_i Q(x_i,(1/m)_{-i})\right)\\
&= &\sum_{i=1}^m \frac{x_i}{1-x_i}\left(\frac{m-1}{m}\right)\\
& &-
 \sum_{i=1}^m x_i\left(\frac{x_i}{1-x_i}\left(\frac{m-1}{m}\right)+(m-1)\frac{1}{m}\right) \\
&= &\sum_{i=1}^m(1-x_i)\frac{x_i}{1-x_i}\left(\frac{m-1}{m}\right)\\
& & -  \left(\sum_{i=1}^m x_i\frac{m-1}{m}\right) \\
& = & 0.
\end{eqnarray*}
This implies that the mixed strategy of the AND player defined by
$F$, is a best response to the mixed strategy of the OR player
defined by $G$. We now show that the mixed strategy of the  OR
player defined by $G$,  is a best response to the mixed strategy of
the AND player defined by $F$.

Recall that $P(x)$, for $x\in [0,1/m]^m$ is the cumulative
probability distribution of the bids of the AND player, and
by the definition of the AND player it equals to
\[
 P(x)= \Pr[\forall i \;bid_i < x_i ]
 = \frac{v-\frac{1}{m}}{v-\min_i\{x_i\}}.
 \]
(Note that, as it should be, under $P$ the support is the set of all identical
bids, i.e., $\forall i\; bid_i=x$. The probability under $P$ of
having a vector $z\leq x$ is $\frac{v-\frac{1}{m}}{v-x}$.)

The
utility function of the OR player is:
\[
U_{OR} = E_{x\sim Q} [u_{or}(x)],
\]
where
\[
u_{or}(x)=v\cdot e(x) - \left(\sum_{i=1}^m x_i
P(x_i,(1/m)_{-i})\right).
\]
where $e(x) = \Pr_P[\exists i \mbox { such that } X_i < x_i]$.

We obtain  that for any $x \in [0,\frac{1}{m}]$ and $i\in [1,m]$
$u_{or}(x_i=x, x_{-i}=0)=v -\frac{1}{m}$ since
\begin{eqnarray*}
u_{or}(x_i=x, x_{-i}=0) &= & \frac{v-(1/m)}{v-x}(v-x)= v-\frac{1}{m}
\end{eqnarray*}
Furthermore, for any $x \in [0,\frac{1}{m}]^m$ we have
$u_{or}(x)\leq u_{or}(y)$, where $y$ keeps only the maximal entry in
$x$ and zeros the rest.
This follows
since given $P$, the probability of winning under $x$ and $y$ is
identical. Clearly the payments under $y$ are at most those under
$x$ (since all the bids in $x$ are at least the bids in $y$). We conclude
that the OR player's strategy is a best
response to the AND player's strategy, and this completes the proof.
\end{proof}

\subsection{The Triangle Game}
\label{sec-triangle}

 We start with a simple case of three single
minded bidders and three items, where each bidder wants a different
set of two items, and has a value of one for this set.

\ignore{ Consider a cycle with three vertices $v_0,v_1,v_2$ and
three edges $e_0=(v_0,v_1)$, $e_1=(v_1,v_2)$, and $e_2=(v_2,v_0)$.
Vertex $v_i$ corresponds to item $i$ and the edge $e_j$ corresponds
to player $j$. The valuation of player $j$ is $1$ for the set of
items $\{ j, (j+1) \mod 3\}$, which are the endpoints of $e_j$, and
$0$ for any other set of items.}

Consider symmetric strategies in which each player bids the same for
the pair of items it wants, namely each player draws their bid $x$
from the same distribution whose cumulative distribution function is
$F(x)$. Assuming $F(x)$ has no atoms then the utility of each player
is
$$
(1-2x)F^2(x)-2xF(x)(1-F(x))=F^2(x)-2xF(x)
$$

\ignore{For $F$ to form an equilibrium this utility should be equal
to some constant $C$ for all $x$ in the support of $F$ and not more
than $C$ for all other $x$'s.}

\ignore{ The solution (thinking of $F(x)$ as the indeterminate
variable) of the quadratic equation
$$
F^2(x)-2xF(x)=C
$$
is
$$
F(x) = x\pm \sqrt{x^2 + C} \ .
$$
Taking the positive root and  $C=0$ we get the solution
 $F(x) = 2x$, which indeed defines a CDF in the range $0\le x\le 1/2$.
(For other values of $C$, $F$ has an atom at $0$.)}

\begin{theorem}
If each player draws an $x$ from $F(x) = 2x$, where $0\le x\le 1/2$,
and bids $x$ on both items, then it is a mixed Nash equilibrium.
\end{theorem}
\begin{proof}
Suppose two of the players play according to $F(x)$ and consider the
best response of the third player. For any value $0 \le x \le 1/2$
if the third player bids $(x,x)$,  his utility is zero. On the other
hand, if it bids $y$ for one item and $z$ for the other then its
utility is $F(y)F(z)\cdot 1 - yF(y) - zF(z) = -2(y - z)^2 \le 0$.
Finally, bidding any number strictly above $1/2$ is dominated by
bidding $1/2$.
\end{proof}

Consider now a generalization of this game where each player is
single minded and is interested only in a particular set of $k$
items for which its utility is $1$. We also make the following
assumptions.
\begin{enumerate}
\item
Exactly $d$ agents are interested in each item.
\item
For any two bidders $i\neq i'$, we have $|S_i\cap S_{i'}|\leq 1$.
(This implies that if we fix a player $i$ and consider its set $S_i$
of $k$ items. The other $(d-1)k$ players who are also interested in
these $k$ items are all different.)
\end{enumerate}

Assume each player $i$ draws the same bid for all items in its set
$S_i$ from the CDF $G(x)$. If $G(x)$ satisfies the equation
\begin{equation}
\label{eqn:eqn2}
 G^{(d-1)k}(x) - kxG^{d-1}(x) = 0
\end{equation}
for all $x$ then the utility of a player is zero for every bid $x$.

One can easily verify that the function
$$
G(x) = (kx)^{\frac{1}{(d-1)(k-1)}} \ ,
$$
satisfies Equation (\ref{eqn:eqn2}) for all $x$. So $G(x)=
(kx)^{\frac{1}{(d-1)(k-1)}}$, $0\le x\le \frac{1}{k}$, forms an
equilibrium for the restricted game, where in the restricted game a
player has to bid the same bid on all the items in his set. The
following shows that even if we do not restrict the players to bid
the same then $G(x)$ is an equilibrium.

\begin{theorem}
\label{thm-single-mind}
If all players draw a bid for all $k$ items
that they want from $G(x) = (kx)^{\frac{1}{(d-1)(k-1)}}$, $0\le x\le
\frac{1}{k}$, then it is a mixed Nash equilibrium.
\end{theorem}

\begin{proof}
Suppose all the players but one play according to $G(x)$ and
consider the best response of the first player. Suppose its bid is
$x_j$ for the $j$th item in $S_i$. Then  its utility is
\[
\Pi_{j=1}^k (kx_j)^{\frac{1}{(k-1)}} - \sum_{j=1}^k x_j
(kx_j)^{\frac{1}{(k-1)}} \ .
\]
We claim that this utility is non-positive for every set of bids
$x_1,\ldots,x_k$. Indeed this follows since
$$
k^{\frac{k}{k-1}}\prod_{j=1}^k x_j^{\frac{1}{(k-1)}} \le
k^{\frac{1}{k-1}} \sum_{j=1}^k x_j^{\frac{k}{(k-1)}}
$$
by the inequality of arithmetic and geometric means:
$$
\sqrt[k]{\prod x_i^{\frac{k}{k-1}}} = \prod_{j=1}^k
x_j^{\frac{1}{(k-1)}} \le \frac{1}{k} \sum_{j=1}^k
x_j^{\frac{k}{(k-1)}} \ .
$$
\end{proof}
%
%

\section{Inefficiency of Mixed Equilibria} \label{sec:lowerbounds}

In this section we use our analysis of the examples
given in the previous section to construct examples where there are large gaps
between the efficiency obtained in a mixed-Nash equilibrium and the optimal efficiency.

We first analyze the AND-OR game with $m$ items, where $v\geq 1/m$, and hence there is no pure Nash equilibrium.
We will analyze the following parameters: value of the OR player is $v=1/\sqrt{m}$ and the value of the AND
player is $1$.

\begin{theorem}
There is a mixed Nash equilibrium in the AND-OR game with the parameters above whose social
welfare is at most $2/\sqrt{m}$.  I.e. for this game we have $PoA \geq \sqrt{m}/2$.
\end{theorem}

\begin{proof}
For the PoA consider the equilibrium of Section \ref{section-and-or}. Assume that the
value of the OR player is $v=1/\sqrt{m}$ and the value of the AND
player is $1$. This implies that the optimal social welfare is $1$.
The probability that the AND player bids $x=0$ is
$\frac{v-1/m}{v-x}= 1-1/\sqrt{m}$. Therefore with probability at
least $ 1-1/\sqrt{m}$ the OR player wins. This implies that the
social welfare is at most $2/\sqrt{m}$
\end{proof}

We now prove the following lemma regarding the support of the AND player.

\begin{lemma}
\label{support-and} In any Nash equilibrium the AND player does not
have in its support any bid vector $b_{and}$ such that $\sum_{i=1}^m
b_{and,i}$ $>1$.
\end{lemma}

\begin{proof}
Assume that there is such a bid vector $b_{and}$.
Since $\sum_{i=1}^m b_{and,i}>1$ the AND player can not get a positive utility, and the
only way it can gain a zero utility is by losing all its non-zero bids.
This implies that for {\em any} bid vector $b_{or}$ of the OR
player, the OR player will win all the items. Therefore
$\sum_{i=1}^m b_{or,i}>1$.  This implies that the revenue of the
auctioneer is larger than $1$ (every time). Since the expected
revenue of the auctioneer is larger than $1$, and the optimal social
welfare is $1$, the sum of the expected utilities of the players has
to be negative. Hence one of the players has an expected negative
utility. This clearly can not occur in equilibrium.
\end{proof}

It turns out that for this example, not only there exist bad
equilibria, but actually all equilibria are bad!

\begin{theorem}
For any Nash equilibrium of the AND-OR game with the parameters above the social welfare is at most $3\sqrt{(\log m)/m}$.
I.e. the $PoS \geq \sqrt{m/\log m}/3$.
\end{theorem}

\ignore{ Let us start with a slightly simpler proof of a weaker
$\Omega(m^{1/3})$ bound.

\begin{proof} (of weaker bound)
Assume we have a Nash equilibrium in which the AND
player wins with probability $\alpha$.
This implies that the expected utility of the OR player $u_{or}$ is at most $(1-\alpha)v$.
Also, the social welfare of the equilibrium is $(1-\alpha)v+\alpha\leq v+\alpha$.

By Lemma \ref{support-and} the AND player never plays a bid $b$ in
which the sum of the bids is larger than $1$.  This implies that the
AND player can have at most $m^{2/3}$ bids which are larger than
$1/m^{2/3}$. Therefore, if the OR player bid $1/m^{2/3}$ on a random
item, it will win with probability at least $1-1/m^{1/3}$. The OR
player utility from such a strategy is at least
$(1-1/m^{1/3})v-1/m^{2/3}$. This implies that in equilibrium,
\[
(1-\alpha)v \geq u_{or}\geq
(1-1/m^{1/3})v-1/m^{2/3}.
\]
For $v=1/m^{1/3}$ it implies that $\alpha
\leq 2/m^{1/3}$.
Therefore the social welfare is at most
$3/m^{1/3}$.
\end{proof}
}

\begin{proof}
Assume we have a Nash equilibrium in which the AND
player wins with probability $\alpha$.
This implies that the expected utility of the OR player $u_{or}$ is at most $(1-\alpha)v$.
Also, the social welfare of the equilibrium is $(1-\alpha)v+\alpha\leq v+\alpha$.

By Lemma \ref{support-and} the AND player never plays a bid $b$ in
which the sum of the bids is larger than $1$.  This implies that the
AND player can have at most half of the bids which are larger than
$2/m$. Therefore, if the OR player bid $2/m$ on $\log m$ random
items, it will win some item with probability at least $1-1/m$. The
OR player utility from such a strategy is at least $(1-1/m)v-(\log
m)/m$. This implies that in equilibrium,
\[
(1-\alpha)v \geq u_{or}\geq
(1-1/m)v-(\log m)/m.
\]
For $v=\sqrt{(\log m)/m}$ it implies that $\alpha
\leq 2\sqrt{(\log m)/m}$.
Therefore the social welfare is at most
$3\sqrt{(\log m)/m}$.
\end{proof}

Finally we study examples in which there are multiple equilibria, and show that they can be far apart from one another:

\begin{theorem}
There is a set of valuations such that in the corresponding
simultaneous first price auction there is an  efficient (pure) Nash
equilibrium, as well as an inefficient one, where the inefficiency
is at least by a factor of $\sqrt{m}/2$. Equivalently, the
corresponding auction has $PoS=1$ but $PoA \geq \sqrt{m}/2$.
\end{theorem}

\begin{proof}
 Consider $m=\ell^2$ items, which are labeled by $(i,j)$ for
$i,j\in[1,\ell]$. Now we analyze $2\ell$ single minded bidder, where
for each $i\in [1,\ell]$ we have a bidder that wants all the items
in $(i,*)$, we call those bidders {\em row bidders}. For each $j\in
[1,\ell]$ we have a bidder that want $(*,j)$, and we call them {\em
column bidders}. All bidders have value $\ell$ for their set. Note
that there is no allocation where both a row and a column players
are satisfied, where a player is satisfied if it is allocated all
the items in his set.
The social optimum value is $\ell^2$ (satisfying all the row bidders
or all the column bidders). In this game there is a Walresian
equilibrium, where the price of each item is $1$. Similarly, there
is a pure Nash equilibrium where all bidders bid $1$ for each item
and we break the ties in favor to all the row players (or
alternatively, to all the column players). This implies that the PoS
is $1$. Note that this game has also a mixed Nash equilibrium
(Section \ref{sec-triangle}, Theorem~\ref{thm-single-mind}). Since
it is a symmetric equilibrium, in which every player bids the same
value on all items, the expected number of satisfied players is at
most $2$ (since the probability of $k$ satisfied players is at most
$2^{-k}$). This implies that the PoA is $\ell/2=\sqrt{m}/2$.
\end{proof}

\section{Approximate Welfare Analysis}

In this section we analyze the Price of Anarchy of the simultaneous
first-price auction. We start with a simple proof of an $O(m)$ upper
bound on the price of anarchy for general valuations. Then we
consider $\beta$-XOS valuations (which are equivalent to
$\beta$-fractionally subadditive valuations) and prove an upper
bound of $2\beta$.
Since subadditive valuations are $O(\log m)$ fractionally
subadditive \cite{Feige09,BhawalkarR11} we also get an upper bound
of $O(\log m)$ on the price of anarchy for subadditive valuations.


Assume that in $OPT$ player $i$ gets set $O_i$ and receives value
$o_i=v_i(O_i)$. Let $k_i$ be $|O_i|$. Let $e_i$ be the expected
value player $i$ gets in an equilibrium and let $u_i$ be the
expected utility of player $i$ in an equilibrium. Let $r_i$ be the
expected sum of payments in equilibrium over all items in $O_i$
(these are not necessarily won by player $i$ in equilibrium).

Denote the total welfare, revenue, and utility in equilibrium by
$SW(eq)$, $REV(eq)$, and $U(eq)$, respectively. By definitions we
have: (1) $SW(eq)=\sum_i e_i$, (2) $SW(OPT)=\sum_i o_i$, (3)
$REV(eq)=\sum_i r_i \leq SW(eq)$, (4) $U(eq)=\sum_i u_i =
SW(eq)-REV(eq)$.

\ignore{
\begin{lemma} \label{PoA-general}
For each $i$, $2u_i \geq o_i - 4k_i r_i$.
\end{lemma}

\begin{proof}
By Markov, with probability of at least $1/2$ the total sum of
prices of items in $O_i$ is at most $2r_i$.  Thus if player $i$ bids
$2r_i$ for each item in $O_i$ (and $0$ elsewhere) he wins all items
with probability of at least $1/2$, getting expected value of at
least $o_i/2$, and paying at most  $2k_i r_i$. Since we were in
equilibrium this utility must be at most $u_i$.
\end{proof}
}

\begin{theorem}
For any set of buyers the PoA is at most $4m$.
\end{theorem}
\begin{proof}
We first show that for each buyer $i$, we have $2u_i \geq o_i - 4k_i
r_i$.

By Markov, with probability of at least $1/2$ the total sum of
prices of items in $O_i$ is at most $2r_i$. Thus if player $i$ bids
$2r_i$ for each item in $O_i$ (and $0$ elsewhere) he wins all items
with probability of at least $1/2$, getting expected value of at
least $o_i/2$, and paying at most $2k_i r_i$. Since we were in
equilibrium this utility must be at most $u_i$. Hence, $2u_i \geq
o_i - 4k_i r_i$.

Summing over all buyers, and bounding $\sum_i k_i \leq m$, we get
that $OPT \leq 2U(eq)+4mRev(eq) \leq 4mSW(eq)$.
\end{proof}

\ignore{

 The same approach gives a better bound for XOS valuations.

\begin{lemma} \label{lem:factor4}
\label{PoA-XOS} Assume that the valuations of all the players are
XOS. Then for each $i$, $2u_i \geq o_i - 4 r_i$.
\end{lemma}

\begin{proof}
Let $f_j$ be the expected price of item $j$. By Markov, with
probability of at least $1/2$ the the price of item $j$ is at most
$2f_j$.  Thus if player $i$ bids $2f_j$ for each item $j\in O_i$
(and $0$ elsewhere) he wins each item with probability of at least
$1/2$. Since the valuations are XOS this gives him an expected value
of at least $o_i/2$, while paying at most $\sum_{j\in O_i}2f_j= 2
r_i$. Since we were in equilibrium this utility must be at most his
utility in equilibrium which is at most $u_i$.
\end{proof}

\begin{theorem} \label{thm:factor4}
For any set of XOS buyers the PoA is at most $4$.
\end{theorem}
\begin{proof}
Summing Lemma~\ref{PoA-XOS} over all $i$'s, we get: $OPT \leq
2U(eq)+4Rev(eq) \leq 4SW(eq)$.
\end{proof}
}

A function $v$ is $\beta$-XOS, if there exists an XOS function $X$
such that for any set $S$ we have $v(S)\geq X(S)\geq v(S)/\beta$,
i.e., if there are numbers $\lambda_{j,l}$, $j\in M$ and $l\in L$,
such that for any set $S$ we have
\[
v(S) \geq \max_{k\in L} \sum_{j\in S} \lambda_{j,k} \geq v(S)/\beta
\]

The equivalence of $\beta$-XOS and  $\beta$ fractionally
sub-additive follows the same proof as in \cite{Feige09}.

\begin{theorem} \label{thm:XOSfactor2}
\label{PoA-XOSbeta} Assume that the valuations of all the players are
$\beta$-XOS.
Then the PoA is $2\beta$.
\end{theorem}

\begin{proof}
Since $v$ is $\beta$-XOS, there is a $k\in L$
such that $\sum_{j\in O_i} \lambda_{j,k}\geq v_i(O_i)/\beta$, and
for any set $S$, we have  $v(S)\geq \sum_{j\in S} \lambda_{j,k}$.
Let $f_j$ be the expected price of item $j$. By Markov inequality,
with probability of at least $1/2$ the  price of item $j$ is at most
$2f_j$.  Consider the deviation where player $i$ bids
$bid_{i,j}=\min\{\lambda_{j,k},2f_j\}$ for each item $j\in O_i$ (and
$0$ elsewhere). Player $i$ wins each item $j$ with probability
$\alpha_j$ and if $bid_{i,j}=2f_j$ then $\alpha_j\geq 1/2$. Let
$S_i$ be the set of item that player $i$ wins with his deviation
bids $bid_{i,j}$. (Note that $S_i$ is a random variable that depends
on the random bids of the other players.)
The expected utility of player $i$ from the deviation is,
\begin{eqnarray*}
E[v_i(S_i)-\sum_{j\in S_i} bid_{i,j}] & \geq &
\sum_{j\in O_i} \alpha_j (\lambda_{j,k} -bid_{i,j})\\
&\geq &
\sum_{j\in O_i} \frac{1}{2} (\lambda_{j,k}-bid_{i,j})\\
&\geq &
\sum_{j\in O_i} \frac{1}{2} (\lambda_{j,k}-2f_j)\\
&\geq & \frac{1}{2\beta}v_i(O_i)- \sum_{j\in O_i} f_j\;.
\end{eqnarray*}
Since player $i$ was playing an equilibrium strategy, we have that
$u_i \geq E[v_i(S_i)-\sum_{j\in S_i} bid_{i,j}]$. Summing over all
players $i$'s, and recalling that $REV(eq)=\sum_{j\in M}f_j$, we
get,
\begin{eqnarray*}
E[SW(eq)]-REV(eq) & = & \sum_{i=1}^n u_i\\
& \geq &
\frac{1}{2\beta}E[SW(OPT)] - REV(eq),
\end{eqnarray*}
which completes the proof.
\end{proof}



\section{Bayesian Price of Anarchy}

In a Bayesian setting
there is a known prior distribution $Q$ over the valuations of the players.
We first
sample $v\sim Q$ and inform each player $i$ his valuation $v_i$.
Following that, each player $i$  draws his bid from the distribution
$D_i(v_i)$, i.e., given a valuation $v_i$ he bids $(b_{i,1}, \ldots
, b_{i,m})\sim D_i(v_i)$.
The distributions  $D(v)=(D_1(v_1), \dots , D_n(v_n))$ are
a Bayesian Nash equilibrium if each $D_i(v_i)$ is a best response of
player $i$, given that its valuation is $v_i$ and the valuations
are drawn from $Q$.

We start with the general case, where the distribution over
valuations is arbitrary and the valuations are also arbitrary. Later
we study product distributions over $\beta$-XOS valuations.

\begin{theorem}
For any prior distribution $Q$ over the players valuations, the
Bayesian PoA is at most $4mn+2$.
\end{theorem}

\begin{proof}
Fix a Bayesian Nash equilibrium $D=(D_1, \dots , D_n)$
as described above. Let $Q_{v_i}$ be the distribution on $v_{-i}$ obtained
by conditioning $Q$ on $v_i$ as the value of player $i$.

Let $u_i(v_i)$ be the expected utility of player $i$ when his
valuation is $v_i$, i.e., $u_i(v_i)=E_{b_i\sim
D_i(v_i)}E_{v_{-i}\sim Q_{v_i}}E_{b_{-i}\sim
D_{-i}}[v_i(S_i)-\sum_{j\in S_i}b_{i,j}]$, where $S_i$ is the set of
items that player $i$ wins with the set of bids $b$. Let $u_i$ be
the expected utility of player $i$, i.e., $E_{v_i\sim Q}[u_i(v_i)]$.

For any valuation $v_i$ for player $i$, consider the following
deviation. Let $Rev(v_i)$ be the expected revenue given that the
valuation of player $i$ is $v_i$, i.e., $Rev(v_i)=E_{v\sim
Q_{v_i}}[\sum_{j=1}^m\max_k b_{k,j}]$. Consider the deviation where player
$i$ bids $2 Rev(v_i)$ on each item $j \in M$. By  Markov
inequality, he will win all the items $M$ with probability at least
$1/2$. Therefore, his utility from the deviation is at least
\[
v_i(M)/2 - 2m Rev(v_i)
\]
Since this is an equilibrium, we have that
\[
u_i(v_i) \geq v_i(M)/2 - 2m Rev(v_i)
\]
Summing over the players and taking the expectation with respect to
$v$,
\begin{eqnarray*}
\sum_{i=1}^n E_v[u_i(v_i)] \geq \sum_{i=1}^n E_v[v_i(M)/2 - 2m
Rev(v_i)]
\end{eqnarray*}
Clearly $\sum_{i=1}^n E_v[u_i(v_i)] \leq E_v(SW(D))$, where
$E_v(SW(D))$ is the expected social welfare of the Bayesian
equilibrium $D$. Also, $\sum_{i=1}^n E_v[v_i(M)]\geq
E_v[SW(OPT(v))]$. Finally, for every player $i$,
$E_{v}[Rev(v_i)]=Rev$, where $Rev$ is the expected revenue.
Therefore,
\begin{eqnarray*}
E_v[SW(D)] \geq E_v[SW(OPT(v))]/2 - 2mn Rev
\end{eqnarray*}
Since $Rev\leq E_v[SW(D)]$, we have that,
\begin{eqnarray*}
(4mn+2)E_v[SW(D)] \geq E_v[SW(OPT(v))]
\end{eqnarray*}
\end{proof}

The following theorem show that the Bayesian PoA is at most $4\beta$
when the valuations are limited to $\beta$-XOS and the distribution
$Q$ over valuations is a product distribution. The proof uses the
ideas presented in \cite{BhawalkarR11}.

\begin{theorem}
For a product distribution $Q$ over  $\beta$-XOS
valuations of the players, the Bayesian PoA is at most $4\beta$.
\end{theorem}

\begin{proof}
Fix a Bayesian Nash equilibrium $D=(D_1, \dots , D_n)$
as described above. Let $Q_{v_i}$ be the distribution on $v_{-i}$ obtained
by conditioning $Q$ on $v_i$ as the value of player $i$.

Consider the following deviation of player $i$, given its valuation
$v_i$. Player $i$ draws $w_{-i}\sim Q_{v_i}$, that is $w_{-i}$ are
random valuations of the other players, conditioned on player $i$
having valuation $v_i$. Player $i$ computes the optimal allocation
$OPT(v_i,w_{-i})$, and in particular his share $OPT_i(v_i,w_{-i})$
in that allocation. Player $i$ bids $2f_j(v_i)$ on each item $j\in
OPT_i(v_i,w_{-i})$ , where $f_j(v_i)$ is the expected maximum bid of
the other players on item $j$ in the equilibrium $D$ conditioned on
player $i$ having valuation $v_i$, i.e., $$f_j(v_i)=E_{w_{-i}\sim
Q_{v_i}}E_{b_{-i}\sim D_{-i}(w_{-i})}[\max_{k\neq i} b_{k,j}]\ .$$
By Markov inequality player $i$ wins each item $j\in OPT_i(v_i,$
$w_{-i})$ with probability at least half. Since $v_i$ is an
$\beta$-XOS valuation, its expected value is at least
$v_i(OPT_i(v_i,w_{-i}))/(2\beta)$ so the utility of player $i$ in
this deviation is at least
\[
E_{w_{-i}\sim Q_{v_i}}[v_i(OPT_i(v_i,w_{-i}))/(2\beta) - \sum_{j\in
OPT_i(v_i,w_{-i})} 2f_j(v_i)]
\]

Let $u_i(v_i)$ be the expected utility of player $i$ when his
valuation is $v_i$, i.e., $u_i(v_i)=E_{b_i\sim
D_i(v_i)}E_{v_{-i}\sim Q_{v_i}}E_{b_{-i}\sim
D_{-i}}[v_i(S_i)-\sum_{j\in S_i}b_{i,j}]$, where $S_i$ is the set of
items that player $i$ wins with the set of bids $b$. Let $u_i$ be
the expected utility of player $i$, i.e., $E_{v_i\sim Q}[u_i(v_i)]$.
We get that,
\begin{eqnarray*}
\lefteqn{u_i(v_i) \geq} \\
& & E_{w_{-i}\sim Q_{v_i}}[v_i(OPT_i(v_i,w_{-i}))/(2\beta) -
\sum_{j\in OPT_i(v_i,w_{-i})} 2f_j(v_i)]
\end{eqnarray*}
Takin the expectation with respect to $v_i$,
\begin{eqnarray*}
\lefteqn{u_i=E_{v_i}[u_i(v_i)]}\\
& \geq& E_{v_i}E_{w_{-i}\sim Q_{v_i}}[v_i(OPT_i(v_i,w_{-i}))/(2\beta)\\
 && - \sum_{j\in OPT_i(v_i,w_{-i})}  2f_j(v_i)]\\
&=& E_{v\sim Q}[v_i(OPT_i(v))/(2\beta)] \\
& &- E_{v\sim Q}[ \sum_{j\in OPT_i(v)}  2f_j(v_i)]\\
&=& E_{v\sim Q}[v_i(OPT_i(v))/(2\beta)] \\
& &-2 E_{v\sim Q}[ \sum_{j\in M}I(j\in OPT_i(v))  f_j(v_i)],
\end{eqnarray*}
where $I(X)$ is the indicator function for the event $X$. Summing
over all the players
\begin{eqnarray*}
\sum_{i=1}^n u_i  & \geq & \sum_{i=1}^n
E_{v\sim Q}[v_i(OPT_i(v))/(2\beta)] \\
& &-2 \sum_{i=1}^n E_{v\sim Q}[ \sum_{j\in M}I(j\in OPT_i(v))
f_j(v_i)]\\
& = & E_{v\sim Q}[SW(OPT(v))/(2\beta)] \\
& &-2 \sum_{j\in M}  E_{v\sim Q}[\sum_{i=1}^n I(j\in OPT_i(v))
f_j(v_i)]
\end{eqnarray*}
%
%
Now we use the fact that the distribution $Q$ over the valuations is
a product distribution. This implies that for any valuation $v_i$,
we have the same value $f_j(v_i)$. Let $price(j)$ be the expected
price of item $j\in M$, i.e., $price(j)=E_{v\sim Q}E_{b\sim D}[$
$\max_{k} b_{k,j}]$.
Since $price(j)\geq f_j(v_i)$ for any buyer $i$ and valuation $v_i$,
\begin{eqnarray*}
\sum_{i=1}^n u_i
& \geq &  E_{v\sim Q}[SW(OPT(v))/(2\beta)] \\
& &-2 \sum_{j\in M} price(j) E_{v\sim Q}[\sum_{i=1}^n I(j\in
OPT_i(v)) ]\\
& =&  E_{v\sim Q}[SW(OPT(v))/(2\beta)] -2 \sum_{j\in M} price(j),
\end{eqnarray*}
where the last equality follows since item $j$ is always assigned to
some buyer, therefore, for any $v$, we have $\sum_{i=1}^n I(j\in
OPT_i(v))=1$.

Let $sw(D)$ be the expected social welfare of the Bayesian Nash $D$.
Note that $\sum_{i=1}^n u_i = sw(D) - \sum_{j\in M} price(j)$.
Therefore,
\begin{eqnarray*}
\lefteqn{sw(D) - \sum_{j\in M} price (j)  \geq }\\ & & E_{v\sim
Q}[SW(OPT(v))/(2\beta)]
 -2 \sum_{j\in M} price(j),
\end{eqnarray*}
which implies that
\begin{eqnarray*}
2sw(D) \geq sw(D) + \sum_{j\in M} price (j) & \geq   E_{v\sim
Q}[SW(OPT(v))/(2\beta)].
\end{eqnarray*}
This implies that the PoA of the Bayesian equilibrium $D$ is at most
$4\beta$.
\end{proof}

\bibliographystyle{plain}
\bibliography{firstprice}

\begin{thebibliography}{1}

\bibitem{ArrowD54}
K.~J. Arrow and G.~Debreu.
\newblock Existence of an equilibrium for a competitive economy.
\newblock {\em Econometrica}, 22:265--290, 1954.

\bibitem{BhawalkarR11}
Kshipra Bhawalkar and Tim Roughgarden.
\newblock Welfare guarantees for combinatorial auctions with item bidding.
\newblock In {\em SODA}, 2011.

\bibitem{BlunrosenNisan07}
L.~Blumrosen and N.~Nisan.
\newblock Combinatorial auctions (a survey).
\newblock In N.~Nisan, T.~Roughgarden, E.~Tardos, and V.~Vazirani, editors,
  {\em Algorithmic Game Theory}. Cambridge University Press, 2007.

\bibitem{fundenberglevine86}
Fundenberg D. and D.~Levine.
\newblock Limit games and limit equilibria.
\newblock {\em Journal of Economic Theory}, 38:261--279, 1986.

\bibitem{Feige09}
U.~Feige.
\newblock On maximizing welfare when utility functions are subadditive.
\newblock {\em SIAM J. Comput.}, 39(1):122--142, 2009.

\bibitem{JacksonsSwinkels05}
Matthew~O. Jackson and Jeroen~M. Swinkels.
\newblock Existence of equilibrium in single and double private value auctions.
\newblock {\em Economertica}, 73(1):93--139, 2005.

\bibitem{dasguptamaskin86}
Dasgupta P. and E.~Maskin.
\newblock The existence of equilibrium in discontinuous economic games.
\newblock {\em Review of Economic Studies}, 53:1--42, 1986.

\bibitem{SimonZame90}
Leo~K. Simon and William~R. Zame.
\newblock Discontinuous games and endogenous sharing rules.
\newblock {\em Econometrica}, 58(4):861--872, 1990.

\end{thebibliography}

\end{document}